\documentclass[a4wide,12pt]{article}

\usepackage[pdftex]{graphicx}
\usepackage{amsmath,amssymb,amsfonts}
\usepackage[textsize=small,textwidth=1.3in]{todonotes}
\usepackage[english]{babel}

\usepackage{hyperref}

\newtheorem{prop}{Proposition}[section]
\newtheorem{thm}[prop]{Theorem}

\newtheorem{cor}[prop]{Corollary}
\newtheorem{Defi}[prop]{Definition}
\newtheorem{Rem}[prop]{Remark}
\newtheorem{Exam}[prop]{Example}
\newtheorem{Exer}[prop]{Exercise}
\newtheorem{Expe}[prop]{Experiment}
\newtheorem{Cons}[prop]{Construction}

\newtheorem{Alg}[prop]{Algorithm}
\newtheorem{Prob}[prop]{Research Problem}
\newenvironment{defi}{\begin{Defi} \rm}{\end{Defi}}
\newenvironment{rem}{\begin{Rem} \rm}{\end{Rem}}
\newenvironment{ex}{\begin{Exam} \rm}{\end{Exam}}

\newenvironment{proof}{\noindent{\bf Proof.}\ }{\hspace*{\fill}$\diamond$\medskip\par}

\newcommand{\fq}{{\mathbb F}_q}

\newcommand{\fqt}{{\mathbb F}_{q^t}}

\newcommand{\ba}{{\bf a}}
\newcommand{\bb}{{\bf b}}
\newcommand{\bc}{{\bf c}}

\newcommand{\bg}{{\bf g}}

\newcommand{\bx}{{\bf x}}
\newcommand{\by}{{\bf y}}

\title{LCD codes over ${\mathbb F}_q $ are as good\\
 as linear codes for $q$ at least four}

\author{Ruud Pellikaan
\footnote{Department of Mathematics and Computing Science, Eindhoven University of Technology, P.O. Box 513, 5600 MB  Eindhoven, The Netherlands. E-mail: g.r.pellikaan@tue.nl}}

\begin{document}

\maketitle

\begin{abstract}
\noindent The hull $H(C)$ of a linear code $C$ is defined by $H(C)=C \cap C^\perp$.
A linear code with a complementary dual (LCD) is a linear code with $H(C)=\{0\}$.
The dimension of the hull of a code is an invariant under permutation equivalence.
For binary and ternary codes the dimension of the hull is also invariant under monomial equivalence
and we show that this invariant is determined by the extended weight enumerator of the code.\\
The hull of a code is not invariant under monomial equivalence if $q\geq 4$.
We show that every ${\mathbb F}_q $-linear code is monomial equivalent with an LCD code in case $q \geq 4$.
The proof uses techniques from Gr\"obner basis theory.
We conclude that if there exists an ${\mathbb F}_q $-linear code with parameters $[n,k,d]_q$ and $q \geq 4$,
then there exists also a LCD code with the same parameters.
Hence this holds for optimal and MDS codes.
In particular there exist LCD codes that are above the Gilbert-Varshamov bound if $q$ is a square and $q\geq 49$
by the existence of such codes that are algebraic geometric.\\
Similar results are obtained with respect to Hermitian LCD codes.
\end{abstract}

\newpage

\section{Introduction}

\noindent The hull $H(C)$ of a linear code $C$ is defined by $H(C)=C \cap C^\perp$.
A linear code with a complementary dual (LCD) is a linear code with $H(C)=\{0\}$.
LCD codes were introduced and studied by Massey \cite{massey:1992} who showed that
these codes are optimal for the two-user binary adder channel (BAC) and that they are asymptotically good.
Sendrier \cite{sendrier:2004} showed that they meet the Gilbert-Varshamov bound and he
\cite{sendrier:1997,sendrier:2000,sendrier:2004,sendrier:2013a} considered the hull of a code
to find a permutation between two equivalent codes
with an application to code-based public key cryptosystems  \cite{sendrier:2002}. \\
Carlet and Guilley gave applications of LCD codes in side-channel attacks (SCA) and fault non-invasive attacks \cite{bringer:2014,carlet:2014,carlet:2016,li:2017},
and may be used for the construction of lattices \cite{hou:2016}. LCD rank metric codes have applications in network coding \cite{braun:2013,kandasamy:2012}.\\
Constructions  of LCD codes are known for the following types of codes:
cyclic \cite{li:2016a,li:2016b,sangwisut:2015,yang:1994}, consta-cyclic \cite{chen:2014,sari:2016} and quasi-cyclic \cite{esmaeili:2009,gueneri:2016},
graphic \cite{leemans:2016} and algebraic geometric \cite{mesnager:2016}.
Optimal and MDS codes that are LCD are considered in
\cite{boonniyoma:2016,carlet:2017,chen:2017,dougherty:2015,galvez:2017,jin:2017,lina:2016,liu:2016,sari:2016,sok:2017}.\\
Apart from the standard or Euclidean inner product,
the notion of LCD codes is researched with respect to the Hermitian inner product in case $q$ is a square \cite{boonniyoma:2016,carlet:2017,li:2017,sharma:2017}.\\
The dimension of the hull of a code is an invariant under permutation equivalence.
For binary and ternary codes the dimension of the hull is also invariant under monomial equivalence
and we show that this invariant is determined by the extended weight enumerator of the code.\\
The hull of a code is not invariant under monomial equivalence if $q\geq 4$.\\
We show that every ${\mathbb F}_q $-linear code is monomial equivalent with an LCD code in case $q \geq 4$,
by means of the theory of Gr\"{o}ber bases generalizing a technique in \cite{carlet:2017}.
We conclude that if there exists an ${\mathbb F}_q $-linear code with parameters $[n,k,d]_q$ and $q \geq 4$,
then there exists also a LCD code with the same parameters.
In particular this holds for optimal and MDS codes.
And there exist LCD codes that are above the Gilbert-Varshamov bound if $q$ is a square and $q\geq 49$
by the existence of such codes that are algebraic geometric.\\
Similar results are obtained with respect to Hermitian LCD codes.

\section{Preliminaries}

\noindent
Let  $q $ be a power of a prime. Then $\fq$ denotes the finite field with $q$ elements.\\
An $[n,k,d]_q$ code is a $k$-dimensional subspace of $\fq^n$ of minimum distance $d$.

\begin{defi}
The {\em standard} or {\em Euclidean inner product} of $\ba$ and $\bb$ in $\fq^n$ is defined by
$$
\ba\cdot \bb = \sum_{i=1}^n a_ib_i.
$$
Now $A \perp B$ if and only if $\ba \cdot \bb = 0$ for all $\ba \in A$ and $\bb \in B$.
The {\em dual} of $C$ is defined by
$$
C^\perp = \{ \ \bx \in \fq^n \ | \ \bx\cdot \bc =0 \mbox{ for all } \bc \in C \ \}.
$$
Let $q$ be an even power of a prime. Define the conjugate $\bar{x}$ of $x \in \fq$ by  $\bar{x}=x^{\sqrt{q}}$.
The {\em Hermitian inner product} of $\ba$ and $\bb$ in $\fq^n$ is defined by
$$
\ba \cdot_H \bb = \sum_{i=1}^n a_i\bar{b}_i.
$$
And $A \perp_H B$ if and only if $\ba \cdot_H \bb = 0$ for all $\ba \in A$ and $\bb \in B$.
The {\em Hermitian dual} of $C$ is defined by
$$
C^{\perp_H} = \{ \ \bx \in \fq^n \ | \ \bx\cdot_H \bc =0 \mbox{ for all } \bc \in C \ \}.
$$
\end{defi}

\begin{defi}Let $C$ be an $\fq $-linear code. \\
The {\em hull} of $C$ is defined by $H(C)=C \cap C^\perp $.
If $H(C) = \{ {\bf 0 } \} $, then $C$ is called a {\em linear} code with {\em complementary dual} or an LCD code.\\
The code $C$ is called {\em Hermitian} LCD if $C \cap C^{\perp_H} = \{ 0 \}$.
\end{defi}

\begin{defi}A {\em permutation} matrix  is a square matrix with zeros and ones, such that in every
row (and in every column) there is exactly one element equal to one.
A {\em diagonal} matrix is a square matrix with zeros outside its diagonal.
A {\em monomial} matrix  is a square matrix such that in every
row (and in every column) there is exactly one nonzero element.
\end{defi}

\begin{rem}A permutation matrix and an invertible diagonal matrix are special monomial matrices.
\end{rem}

\begin{defi}Let $C_1$ and $C_2$ be $\fq $-linear codes in $\fq^n$.
Then $C_1$ is called {\em equivalent} or {\em permutational equivalent} to $C_2$,
and is denoted by $C_1\equiv C_2$ if there exists a permutation matrix $\Pi$ such that $\Pi(C_1)=C_2$.
The code $C_1$ is {\em scalar equivalent} with $C_2$ if there is an invertible diagonal matrix $D$
such that $D(C_1)=C_2$.
And $C_1$ is called {\em monomial equivalent} to $C_2$,
denoted by $C_1\cong C_2$ if there exists a monomial matrix $M$ such that $M(C_1)=C_2$.
\end{defi}

\begin{defi}Let $\bx$, $\by  \in \fq ^n$.
Then the {\em star product} is defined by
$$
\bx *\by =(x_1y_1,\ldots ,x_ny_n).
$$
Let $\bx \in \fq ^n$ with nonzero entries, then $\bx^{-1}=(x_1^{-1},\ldots ,x_n^{-1})$.\\
Let $C \subseteq \fq^n$. Then $\bx *C = \{ \ \bx *\bc \ | \ \bc \in C \ \}$.\\
Let $G$ be an $k\times n$ matrix with $i$-th row $\bg_i$.
Then $\bx*G$ is the $k\times n$ matrix with $i$-th row $\bx*\bg_i$.
\end{defi}

\begin{rem}Note that $C_1$ and $C_2$ are scalar equivalent if and only if
there exists an $\bx $ with nonzero entries such that $C_2=\bx * C_1$.\\
Let $C$ be an $\fq$-linear code of length $n$ and $\bx \in \fq ^n$ with nonzero entries.\\
If $G$ is a generator matrix and $H$ a parity check matrix of $C$,
then $\bx*G$ is a generator matrix and $\bx ^{-1}*H$ is a parity check matrix of $\bx*C$.
\end{rem}

\section{The dimension of the hull of a code}
The dimension of the hull of a code is considered in \cite{sendrier:1997}.\\
The dimension of $H(C)$ is an invariant of permutational equivalent codes but not for monomial equivalent codes,
except for $q=2$ and $q=3$, where it can be expressed in terms of the Tutte polynomial of the matroid of $C$.\\
Let $M(C)$ be the matroid of the code $C$ and $t_{M(C)}(X,Y)$ the two variable Tutte polynomial of the matroid.
For a binary code $C$ we have
$$
|t_{M(C)}(-1,-1)|= 2^{\dim H(C)}.
$$
See \cite{rosenstiehl:1978} for graph codes and \cite[Proposition 6.5.4]{brylawski:1992} for arbitrary codes.\\
For a ternary code $C$ we have
$$
|t_{M(C)}(j,j^2)| = (\sqrt{3})^{\dim H(C)},
$$
where $j=e^{2\pi i/3}$ is a primitive third root of unity of the complex numbers and $|z|$ is the modulus of the complex number $z$. See \cite{jaeger:1989}.\\
$W_C(X,Y)$, the weight enumerator  of an $\fq $-linear code $C$ can be expressed in terms of $t_{M(C)}(X,Y)$ as follows:
\[ W_C(X,Y)=(X-Y)^kY^{n-k}\ t_C\left(\frac{X+(q-1)Y}{X-Y},\frac{X}{Y}\right)\; . \]
See Greene \cite{greene:1976}. But the converse is not true.

\begin{ex}\label{ex-1}
The dimension of the hull of a binary code $C$ cannot be expressed in terms of $W_C(X,Y)$.
Consider the two codes $C_1$ and $C_2$ of \cite[Example 1]{simonis:1994} with generator matrices
$G_1=(I_3|I_3)$ and $G_2=(I_3|J_3)$, where $J_3$ is the $3\times 3$ all ones matrix.
Both have the same weight enumerator:
$$
W_{C_1}(X,Y)=W_{C_2}(X,Y)=X^6+3X^4Y^2+3X^2Y^4+Y^6.
$$
Furthermore $C_1$ is self dual, so $\dim H(C_1)=3$.
But $H(C_2)$ is generated by the all ones vector, so $\dim H(C_2)=1$.
\end{ex}

\begin{rem}
Let $t_{C}(X,Y)=t_{M(C)}(X,Y)$. Then
$t_{C}(X,Y)$ can be expressed in $W_C(X,Y,T)$, the extended weight enumerator  of $C$ and vice versa. See \cite[Theorems 5.9 and 5.10]{jurrius:2013}.
\[ W_C(X,Y,T)=(X-Y)^kY^{n-k}\ t_C\left(\frac{X+(T-1)Y}{X-Y},\frac{X}{Y}\right)\; . \]
\[ t_C(X,Y)=Y^n(Y-1)^{-k}W_C(1,Y^{-1},(X-1)(Y-1))\; . \]
So we have in the binary case
\[ |W_C(1,-1,4)|=2^k|t_C(-1,-1)|=2^{k+h}. \]
where $k$ and $h$ are the dimensions of $C$ and $H(C)$, respectively.
\end{rem}

\begin{ex}\label{ex-2}For the extended weight enumerator $W_C(X,Y,T)-X^n $ is divisible by $T-1$.
Define $\bar{W}_C(X,Y,T)=(W_C(X,Y,T)-X^n)/(T-1)$.\\
In the previous Example \ref{ex-1} we have that
$$
\bar{W}_{C_1}(X,Y,T)=3X^4Y^2+3(T-1)X^2Y^4+(T-1)^2Y^6 
$$
So $W_{C_1}(1,-1,4)=2^6$. And $\bar{W}_{C_2}(X,Y,T)$ is equal to
$$
3X^4Y^2+(T-2)X^3Y^3+3X^2Y^4+3(T-2)XY^5+(T^2-3T+3)Y^6.
$$
So $W_{C_2}(1,-1,4)=2^4$.
\end{ex}

\begin{rem}
For the ternary case  we have
\[ t_C(j,j^2)=
j^{2n}(j^2-1)^{-k}W_C(1,j) \; , \]
since $j^{-2}=j$ and $(j-1)(j^2-1)=3$.
So the dimension of the hull of a ternary code $C$ can be expressed in terms of $W_C(X,Y)$.
Furthermore
\[ W_C(1,j)= a_0 + a_1j + a_2j^2\; , \]
where $a_0=\sum_{i=0}^{n/3} A_{3i}$, $a_1=\sum_{i=0}^{n/3} A_{3i+1}$ and $a_2=\sum_{i=0}^{n/3} A_{3i+2}$.
\end{rem}

\section{Existence of LCD codes}

\begin{prop}\label{prop-GGT}Let $C$ be an $\fq$-linear $[n,k]$ code with generator matrix $G$.
Let $h$ be the dimension of $H(C)$ and $r=k-h$. Then $C$ has a generator matrix $G_0$ such that
$$
G_0G_0^T=
\left(
\begin{array}{c|c}
O_{h\times h}&O_{h\times r} \\
\hline
O_{r \times h} &P \\
\end{array}
\right),
$$
where $O_{l\times m}$ is the all zeros $l\times m$ matrix and $P$ is an invertible $r \times r$ matrix.
Furthermore the rank of $GG^T$ is $r$ for every generator matrix $G$ of $C$.
\end{prop}

\begin{proof}Let $\bg_1, \ldots \bg_h$ be a basis of $H(C)$. Now $H(C)$ is a sub space of $C$,
so we can extend $\bg_1, \ldots \bg_h$ to a basis $\bg_1, \ldots \bg_k$ of $C$.
Let $G_0$ be the $k \times n$ matrix with $i$-th row $\bg_i$. Then $G_0$ is a generator matrix of $C$ and
$G_0G_0^T$ is a $k \times k$ matrix with $\bg_i \cdot \bg_j$ at position $(i,j)$.
Now  $\bg_i \cdot \bg_j=0$ if $i\leq h$ or $j\leq h$,
since $\bg_i, \bg_j \in C$ for all $1 \leq i,j \leq n$ and $\bg_i \in C^\perp$ if $i\leq h$ and $\bg_j \in C^\perp$ if $j\leq h$.
Therefore $G_0G_0^T$ has the form as stated in the Proposition.\\
Now suppose that $P$ is not invertible. Then the rows of $P$ are dependent and
there exists an invertible $r\times r$ matrix $M$ such that
the first row of $MP$ consists of zeros. But then also the the first row of $MPM^T$ consists of zeros.
Let $G_f$ be the submatrix of $G$ consisting of the first $h$ rows of $G$,
and $G_l$ the submatrix of $G$ consisting of the last $r$ rows of $G$.
Let $G_1$ be the $k \times n$ matrix with $G_f$ in the first $h$ rows and $MG_l$ in the last $r$ rows.
Then $G_1$ is also a generator matrix of $C$, since $M$ is invertible. Moreover
$$
G_1G_1^T=
\left(
\begin{array}{c|c}
O_{h\times h}&O_{h\times r} \\
\hline
O_{r \times h} &MPM^T \\
\end{array}
\right).
$$
Let $\bg'_i$ be the $i$-th row of $G_1$. Then $\bg'_i \cdot \bg'_j=0$ if $i\leq h+1$ and $j\leq k$,
since the first row of $MPM^T$ consists of zeroes. Hence $H(C)$ consists of $h+1$ independent elements
$\bg'_1, \ldots \bg'_h, \bg'_{h+1}$ which is a contradiction.
Therefore $P$ is invertible and $G_0G_0^T$ has rank $r$.\\
Now suppose that $G$ is a generator matrix of $C$. Then there exists an invertible $k \times k$ matrix $N$
such that $G=NG_0$. Hence $GG^T=NG_0G_0^TN^T$. The rank of $G_0G_0^T$ is $r$ and this rank does not change
by elementary row operations, that is by multiplication on the left by $N$
and also not by elementary column operations, that is by multiplication on the right by $N^T$.
Therefore $GG^T$ has also rank $r$.\\
Remark that with the theory of symmetric forms one can show a more restrictive normal form of $P$. See \cite[Chap. XIV]{lang:1965}.
\end{proof}

\begin{cor}\label{cor-GGT}Let $C$ be an $\fq$-linear $[n,k]$ code with generator matrix $G$.\\
Then $C$ is an LCD code if and only if $GG^T$ is nonsingular. 
\end{cor}

\begin{proof}This is \cite[Proposition 1]{massey:1992} and a consequence of Proposition \ref{prop-GGT}.
\end{proof}

\begin{cor}\label{cor-HGGT}Let $C$ be an $\fq$-linear $[n,k]$ code with generator matrix $G$.\\
Then $C$ is a Hermitian LCD code if and only if $G\bar{G}^T$ is nonsingular. 
\end{cor}

\begin{proof}This is \cite[Theorem 3.4]{boonniyoma:2016}.
\end{proof}

\noindent
The multivariate notation is used for polynomials, that means that $X$ is an abbreviation of $(X_1, \ldots , X_n)$ and
the polynomial $f(X)$ of $f(X_1, \ldots , X_n)$.

\begin{prop}Let $f(X)$ be a nonzero polynomial of $\fq [X_1,\ldots ,X_n]$
such that the degree of $f(X)$ with respect to $X_j$ is at most $q-1$ for all $j$.
Then there exists a $\bx \in \fq ^n$ such that $f(\bx)\not=0$.
\end{prop}

\begin{proof}See \cite[V, \S 4, Theorem 5]{lang:1965}.
\end{proof}

\begin{prop}\label{prop-fX=0}Let $f(X)$ be a nonzero polynomial of $\fq [X_1,\ldots ,X_n]$
such that the degree of $f(X)$ with respect to $X_j$ is at most $q-2$ for all $j$.
Then there exists a $ \bx \in (\fq \setminus \{ 0 \}) ^n$ such that $f (\bx )\not=0$.
\end{prop}

\begin{proof}For the theory of Gr\"{o}bner bases and the concepts of a footprint and delta set we refer to  \cite{adams:1994,cox:2005,cox:2015,hoholdt:19998}.\\\
Let $f(X)$ be a nonzero polynomial of $\fq [X_1,\ldots ,X_n]$
such that the degree of $f(X)$ with respect to $X_j$ is at most $q-2$ for all $j$.\\
Consider the ideal $I_{q,n}$ in $ \fq [X_1,\ldots ,X_n]$ generated by the $X_j^{q-1}-1$, $j=1, \ldots ,n$.
Then $(\fq \setminus \{ 0 \}) ^n$ is the zero set of $I_{q,n}$, and $I_{q,n}$
is the vanishing ideal of $(\fq \setminus \{ 0 \}) ^n$.
The footprint or delta set of the $X_j^{q-1}-1$, $j=1, \ldots ,n$ is equal to $(\{ 0, 1, \ldots ,q-2\})^n$ and contains
the delta set of the ideal of  $I_{q,n}$. The delta set of $I_{q,n}$ is finite and has size $(q-1)^n$,
the size of the zero set of $I_{q,n}$ and it is equal to the size of delta set of the $X_j^{q-1}-1$, $j=1, \ldots ,n$ .
Hence the delta set of $I_{q,n}$ is equal to $(\{ 0, 1, \ldots ,q-2\})^n$ and  $\{ X_j^{q-1}-1 | j=1, \ldots ,n \}$ is a
Gr\"obner basis of $I_{q,n}$ with respect tot the total degree lex order. \\
Now $f(X)$ is a nonzero element of $\fq [X_1,\ldots ,X_n]$.
The exponent of the leading monomial of $f(X)$ is in the delta set of $I_{q,n}$,
since the degree of $f(X)$ with respect to $X_j$ is at most $q-2$ for all $j$.
Therefore $ f(X)\not\in I_{q,n}$.\\
Suppose that $f (\bx ) =0$ for all $ \bx \in (\fq \setminus \{ 0 \}) ^n$.
Then $f(X)$ is in the vanishing ideal of $(\fq \setminus \{ 0 \}) ^n$, which is $I_{q,n}$. So $ f(X)\in I_{q,n}$. This is a contradiction.\\
Therefore there exists a $ \bx \in (\fq \setminus \{ 0 \}) ^n$ such that $f (\bx )\not=0$.
\end{proof}

\begin{thm}\label{thm-LCD}Let $C$ be an $\fq$-linear code with $q\geq 4$. Then $C$ is monomial equivalent with an LCD code.
\end{thm}

\begin{proof}Let $C$ be an $\fq$-linear code with $q\geq 4$. \\
Without loss of generality we may assume that $C$ has a generator matrix of the form $G=(I_k|B)$.
Let $\bx =(x_1, \ldots , x_k)$ be an $k$-tuple of nonzero elements of $\fq $.
Let $D(\bx)$ be the diagonal matrix with $\bx$ on its diagonal.
Let $(D(\bx)|B)$ be the generator matrix of the code $C_{\bx}$. Then $C_{\bx}$ is monomial equivalent with $C$.
Let $X=(X_1, \ldots , X_k)$. Now define
$$
f(X)=\det ((D(X)|B)(D(X)|B)^T).
$$
Then
$$
f(X) = \det (D(X_1^2, \ldots , X_k^2)+BB^T)
$$
Hence $f(X)$ is a polynomial in the variables $X_1, \ldots , X_k$ and the degree
of $f(X)$ with respect to $X _i$ is $2$ for all $i$, which is at most $q-2$, since $q\geq 4$.
The leading term of $f(X)$ with respect to the total degree lex order is $X_1^2 \cdots X_k^2$.
So $f(X)$ is a nonzero polynomial. Therefore $f(\bx)\not=0$ for some
$\bx \in (\fq \setminus \{ 0 \}) ^k$ by Proposition \ref{prop-fX=0}.  \\
Hence $C_{\bx}$ is an LCD code for this choice of $\bx$ by Corollary \ref{cor-GGT}.
\end{proof}

\begin{thm}\label{thm-HLCD}Let $C$ be an $\fq$-linear code with $q$ a square and $q> 4$.
Then $C$ is monomial equivalent with a Hermitian LCD code.
\end{thm}

\begin{proof}The proof is similar to the proof of Theorem \ref{thm-LCD}.
Let $C$ be an $\fq$-linear code with $q$ a square and $q> 4$. \\
Without loss of generality we may assume that $C$ has a generator matrix of the form $G=(I_k|B)$.
Let $\bx =(x_1, \ldots , x_k)$ be an $k$-tuple of nonzero elements of $\fq $.
Let $D(\bx)$ be the diagonal matrix with $\bx$ on its diagonal.
Let $(D(\bx)|B)$ be the generator matrix of the code $C_{\bx}$. Then $C_{\bx}$ is monomial equivalent with $C$.
Let $X=(X_1, \ldots , X_k)$. Now define
$$
g(X)=\det ((D(X)|B)(D(X^{\sqrt{q}})|\bar{B})^T).
$$
Then
$$
g(X) = \det (D(X_1^{\sqrt{q}+1}, \ldots , X_k^{\sqrt{q}+1})+B\bar{B}^T)
$$
Hence $g(X)$ is a polynomial in the variables $X_1, \ldots , X_k$ and the degree
of $g(X)$ with respect to $X _i$ is $\sqrt{q}+1$ for all $i$, which is at most $q-2$, since $q> 4$.
The leading term of $f(X)$ with respect to the total degree lex order is $X_1^{\sqrt{q}+1} \cdots X_k^{\sqrt{q}+1}$.
So $f(X)$ is a nonzero polynomial. Therefore $f(\bx)\not=0$ for some
$\bx \in (\fq \setminus \{ 0 \}) ^k$ by Proposition \ref{prop-fX=0}.  \\
Hence $C_{\bx}$ is a Hermitian LCD code for this choice of $\bx$ by Corollary \ref{cor-HGGT}.
\end{proof}

\begin{rem}Both Theorems \ref{thm-LCD} and \ref{thm-HLCD} are generalizations of \cite[Proposition 3.1]{carlet:2017} and \cite[Proposition 4.1]{carlet:2017}, respectively.
As a result one can now state that the existence of an MDS code with given parameters $[n,k,n-k+1]_q$ is equivalent to the existence of a LCD MDS code with the same parameters if $q\geq 4$,
and to the existence of  Hermitian LCD MDS codes if $q > 4$.
And similar statements hold for optimal codes. See
\cite{boonniyoma:2016,carlet:2017,chen:2017,dougherty:2015,galvez:2017,jin:2017,lina:2016,liu:2016,sari:2016,sok:2017}.
\end{rem}

\begin{rem}There exist (Hermitian) LCD codes that are above the Gilbert-Varshamov (GV) bound if $q$ is a square and $q\geq 49$, by Theorem \ref{thm-LCD}
and since there are algebraic geometric codes above the GV bound. See \cite{bassa:2015,tsfasman:1982}.
\end{rem}

\section{Conclusion}
The dimension of the hull of a code is determined by the extended weight enumerator of the code in case $q=2$ or $q=3$.
We show that every ${\mathbb F}_q $-linear code is monomial equivalent with an LCD code in case $q \geq 4$,
and similarly with an Hermitian LCD in case $q$ is an even power of a prime and $q > 4$.

\bibliographystyle{plain}

\end{document}